\documentclass[10pt,conference]{IEEEtran}
\usepackage{amsmath}
\usepackage{amssymb}
\usepackage{amsfonts}
\usepackage{amscd}
\usepackage{mathrsfs}
\usepackage[final]{graphicx}
\usepackage{color}
\usepackage{url}
\usepackage{verbatim}

\newtheorem{theorem}{Theorem}
\newtheorem{lemma}{Lemma}

\newtheorem{defn}{Definition}
\newtheorem{eg}{Example}
\newtheorem{proposition}{Proposition}

\begin{document}

\title{STBCs from Representation of Extended Clifford Algebras}

\author{
\authorblockN{G. Susinder Rajan}
\authorblockA{ECE Department \\
Indian Institute of Science\\
Bangalore 560012, India\\
susinder@ece.iisc.ernet.in}
\and
\authorblockN{B. Sundar Rajan}
\authorblockA{ECE Department\\
Indian Institute of Science\\
Bangalore 560012, India\\
bsrajan@ece.iisc.ernet.in}
}

%

\maketitle

\begin{abstract}
A set of sufficient conditions to construct $\lambda$-real symbol Maximum Likelihood (ML) decodable STBCs have recently been provided by Karmakar et al. STBCs satisfying these sufficient conditions were named as Clifford Unitary Weight (CUW) codes. In this paper, the maximal rate (as measured in complex symbols per channel use) of CUW codes for $\lambda=2^a,a\in\mathbb{N}$ is obtained using tools from representation theory. Two algebraic constructions of codes achieving this maximal rate are also provided. One of the constructions is obtained using linear representation of finite groups whereas the other construction is based on the concept of right module algebra over non-commutative rings. To the knowledge of the authors, this is the first paper in which matrices over non-commutative rings is used to construct STBCs. An algebraic explanation is provided for the 'ABBA' construction first proposed by Tirkkonen et al and the tensor product construction proposed by Karmakar et al. Furthermore, it is established that the $4$ transmit antenna STBC originally proposed by Tirkkonen et al based on the ABBA construction is actually a single complex symbol ML decodable code if the design variables are permuted and signal sets of appropriate dimensions are chosen. 
\end{abstract}

\section{Introduction}
\label{sec1}
Space-Time Coding is a coding technique for the Multiple Input Multiple Output (MIMO) channel to exploit the spatial diversity that is available in the physical channel. However, Maximum Likelihood (ML) decoding of general STBCs becomes computationally prohibitive especially for large number of transmit antennas. For this reason, Complex Orthogonal Designs (CODs) \cite{TJC,TiH} and quasi-orthogonal designs \cite{TBH,Jaf} were studied. In \cite{KhR}, a class of codes called Co-ordinate Interleaved Orthogonal Designs (CIODs) were proposed which were single complex ML decodable and also had higher rate than CODs. Further, necessary and sufficient conditions to obtain single complex symbol ML decodable codes were provided. Recently in \cite{KaR1,KaR3,KaR2,YGT}, this notion of reduced ML decoding complexity has been put in the more general framework of $g$-group ML decodable or $\lambda$-real symbol ML decodable STBCs. In \cite{KaR1,KaR3}, a special class of multi-symbol ML decodable codes called Clifford Unitary Weight (CUW) codes was introduced. These STBCs were based on a set of sufficient conditions for $g$-group ML decodability. The authors of \cite{KaR1,KaR3} also provided an explicit construction of CUW codes by performing matrix manipulations on the representation matrices of Clifford algebras. Recently in \cite{RaR_isit_p3}, an algebraic framework based on 'Extended Clifford Algebras' (Definition \ref{def_ECA}) has been introduced for constructing CUW codes. Using the framework of Extended Clifford Algebras, $4$-group ML decodable distributed space-time codes for cooperative diversity have been obtained in \cite{RaR_isit_p3}. Distributed space-time codes need to satisfy certain special conditions and these issues have been addressed in \cite{RaR_isit_p3}. Interestingly, using 'Extended Clifford Algebras' and some special signal set constructions, distributed differential space-time codes have been obtained in \cite{RaR_isit_p2} for application in cooperative wireless networks with no channel state information. 

In this paper, we focus only on STBCs for the colocated MIMO channel (no co-operation case). Though few CUW code constructions are available in the literature \cite{KaR1,KaR3,KaR2,YGT}, the maximal rate of CUW codes is still not addressed. For the solitary case of single complex symbol ML  decodable codes alone, a solution is available in \cite{KaR1}. In this paper, we obtain the maximal rate of CUW codes using representation theory of finite groups and the algebraic framework of Extended Clifford Algebras.

The main contributions of this paper are as follows.
\begin{itemize}
\item Using tools from representation theory, the maximal rate of CUW codes is found for the case of $\lambda=2^a,a\in\mathbb{N}$. 
\item Two algebraic constructions of CUW codes with maximal rates are then provided. One of the constructions (in the proof of Theorem \ref{thm_maxrate}) is obtained using linear representation of finite groups whereas the other construction (in Section \ref{sec5}) is based on the concept of right module algebra over non-commutative rings. To our knowledge, this is the first paper in which matrices over non-commutative rings is used to construct STBCs.  
\item An algebraic explanation is given for the 'ABBA' construction originally proposed in \cite{TBH} and the recently proposed tensor product construction in \cite{KaR2} which are also of maximal rate.
\item It is shown that the STBC for $4$ transmit antennas proposed in \cite{TBH} based on the ABBA construction is actually a single complex symbol ML decodable code if the design variables are permuted and appropriate signal sets are chosen. Previously this design has been treated to be a double complex symbol ML decodable code.
\end{itemize}

The rest of the paper is organized as follows:  
In Section \ref{sec2}, the problem statement is described. In Section \ref{sec3}, an introduction to the algebraic tools used in this paper is given. The maximal rate of CUW codes and an algebraic explanation for the tensor product construction in \cite{KaR2} is given in Section \ref{sec4}. An algebraic explanation for the ABBA construction is given in Section \ref{sec5} and discussions on further work comprise Section \ref{sec6}.

\noindent
\textbf{Notation:}
For a complex matrix $A$, $A_{I}$ denotes the real matrix obtained by taking the real parts of all the entries of $A$ and $A_{Q}$ denotes the real matrix obtained by taking the imaginary parts of all the entries of $A$. If $M$ is a module over a ring $\mathcal{B}$, then $End_{\mathcal{B}}(M)$ denotes the set of all $\mathcal{B}$ linear maps from $M$ to $M$. For a vector space $V$, $GL(V)$ denotes the set of invertible linear transformations from $V$ to itself.

\section{Problem Statement}
\label{sec2}
In this section, we briefly introduce CUW codes and formulate the problem statement. We refer the readers to \cite{KaR3} for a detailed explanation.

Consider a $N_t\times N_t$ linear design or linear STBC $S$ in $K$ real variables $x_1,x_2,\dots,x_K$. It can be written as
\begin{equation}
S=\sum_{i=1}^{K}x_iA_i
\end{equation} 
where, $A_i\in\mathbb{C}^{N_t\times N_t}$ is called the weight matrix corresponding to the real variable $x_i$. Let $K=g\lambda$ where, $\lambda$ denotes the maximum number of real variables we would like to jointly decode and $g$ is the number of groups into which the $K$ real variables will be partitioned into. Let us list down the weight matrices in the form of an array as follows. 
\begin{center}
\begin{tabular}{c|ccc}
$A_1$ & $A_{\lambda+1}$ & \dots & $A_{(g-1)\lambda+1}$\\
\hline
$A_2$ & $A_{\lambda+2}$ & \dots & $A_{(g-1)\lambda+2}$\\
$\vdots$ & $\vdots$ & $\ddots$ & $\vdots$\\
$A_{\lambda}$ & $A_{2\lambda}$ & \dots & $A_{K}$
\end{tabular}
\end{center}  
Without loss of generality we shall consider the following partition. All the weight matrices in one column will belong to one group. Moreover, we assume all the weight matrices to be unitary and furthermore set $A_1=I$. Then, it has been shown in \cite{KaR1,KaR3} that to obtain a a $g$-group ML decodable linear design as required, it is sufficient to design the matrices in the first row and the first column such that they satisfy the following conditions. 
\begin{enumerate}
\item The matrices in the first row except $A_1=I$ should form a Hurwitz-Radon family. In other words, all the matrices in the first row except $A_1=I$ should square to $-I$ and should pair-wise anti-commute among themselves. 
\item The matrices in the first column should square to $I$ and should commute with all the matrices in the first row and first column.
\end{enumerate}
Once such a set of matrices is obtained, the matrix in the $i$-th row and the $j$-th column can be filled up by multiplying $A_i$ and $A_{(j-1)\lambda+1}$. Such a set of weight matrices will result in a linear design which will be $\lambda$-real symbol ML decodable. In other words, the ML decoding metric will split into a sum of $g$-terms such that each term is a function of atmost $\lambda$ real variables. Thus ML decoding can be performed by jointly decoding atmost $\lambda$ real variables. Linear STBCs satisfying the sufficient conditions 1) and 2) stated above are called CUW codes. 

In this paper, we are interested in the maximum rate $R=\frac{K}{2N_t}$ (in complex symbols per channel use) of CUW codes. This problem can be formally stated in many equivalent ways. Some of them are listed as follows.
\begin{enumerate}
\item Given $\lambda$ and $N_t$ what is the maximum value of $R$?
\item Given $g$ and $N_t$ what is the maximum value of $R$?
\item \label{question} Given $g$ and $\lambda$ what is the minimum value of $N_t$? 
\end{enumerate}
For $\lambda=2$, the solution to the first question is reported in \cite{KaR1}. In Section \ref{sec4} of this paper the solution to question number \ref{question}) for $\lambda=2^a,a\in\mathbb{N}$ is provided.

\section{Algebraic tools}
\label{sec3}
In this Section, we briefly introduce the algebraic framework first proposed in \cite{RaR_isit_p3} to construct CUW codes.  Then we recall few definitions from algebra and restate the problem described in the previous section in algebraic terms.
\begin{defn}
A nonempty set $\mathcal{B}$ equipped with two binary operations called addition and multiplication denoted by $+$ and $.$ is called a ring denoted by $(\mathcal{B},+,.)$ if 
\begin{enumerate}
\item $(\mathcal{B},+)$ is a Abelian group
\item $(\mathcal{B},.)$ is a monoid with multiplicative identity $1$
\item $x.(y+z)=x.y+x.z,~\forall~x,y,z\in\mathcal{B}$
\item $(x+y).z=x.z+y.z,~\forall~x,y,z\in\mathcal{B}$
\end{enumerate}
\end{defn}
\begin{defn}
A nonempty set $\mathcal{A}$ equipped with two binary operations called addition and multiplication denoted by $+$ and $.$ is called a right module algebra over a ring $\cal{B}$ if 
\begin{enumerate}
\item $(\mathcal{A},+,.)$ is a ring
\item There is a map $(x,\alpha)\to x\alpha$ of $\mathcal{A}\times\mathcal{B}$ into $\mathcal{A}$ satisfying the following for $\alpha,\beta,1\in\mathbb{B}$ and $x,y\in\mathbb{A}$.
\begin{equation}
\begin{array}{c}
(x+y)\alpha=x\alpha+y\alpha\\
x(\alpha+\beta)=x\alpha+x\beta\\
x(\alpha\beta)=(x\alpha)\beta\\
x1=x
\end{array}
\end{equation}
\end{enumerate}
\end{defn}
Note that in the standard mathematical literature (for eg \cite{Jac}), algebra is defined over a field. Since our definition differs from the definition in \cite{Jac}, we have given the name 'right module algebra' in order to distinguish it from the concept of algebra over a field.
\begin{defn}\cite{RaR_isit_p3}
\label{def_ECA}
Let $L=2^a,a\in \mathbb{N}$. An Extended Clifford algebra denoted by $\mathbb{A}_n^L$ is the associative algebra over $\mathbb{R}$~ generated by $n+a$ objects $\gamma_k,\ k=1,\dots,n$ and $\delta_i,\ i=1,\dots,a$ which satisfy the following relations:
\begin{itemize}
\item $\gamma_k^2=-1,\ \forall\ k=1,\dots,n$
\item $\gamma_k\gamma_j=-\gamma_j\gamma_k,\ \forall\ k\neq j$
\item $\delta_k^2=1,\ \forall k=1,\dots,a$
\item $\delta_k\delta_j=\delta_j\delta_k,\ \forall\ 1\leq k,j\leq a$
\item $\delta_k\gamma_j=\gamma_j\delta_k,\ \forall\ 1\leq k\leq a, 1\leq j\leq n$
\end{itemize}
\end{defn}

It is clear that the classical Clifford algebra, denoted by $Cliff_n$, is obtained when only the first two relations are satisfied and there are no $\delta_i$. $Cliff_n$ is a sub-algebra of $\mathbb{A}_n^L$. Let $\mathscr{B}_n$ be the natural $\mathbb{R}$~ basis for this sub-algebra. Then a natural $\mathbb{R}$~ basis for $\mathbb{A}_n^L$ is
\begin{equation}
\begin{array}{rl}
\mathscr{B}_n^L=&\mathscr{B}_n\cup\left\{\mathscr{B}_n\delta_i|i=1,\dots,a\right\}\\
&\bigcup_{m=2}^{a}\mathscr{B}_n\left\{\prod_{i=1}^{m}\delta_{k_i}|1\leq k_i\leq k_{i+1}\leq a\right\}
\end{array}
\end{equation}
\noindent
where
\begin{equation}
\begin{array}{rl}
\mathscr{B}_n=&\left\{1\right\}\bigcup\left\{\gamma_i|i=1,\dots,n\right\}\\
&\bigcup_{m=2}^{n}\left\{\prod_{i=1}^{m}\gamma_{k_i}|1\leq k_i\leq k_{i+1}\leq n\right\}.
\end{array}
\end{equation}
The algebra $\mathbb{A}_n^L$ over $\mathbb{R}$ can also be viewed as a right module algebra over the base ring $Cliff_n$. We will use this fact later in Section \ref{sec5}.

From the defining relations of the generators of the Extended Clifford Algebra, it can be observed that the symbols $1$, $\gamma_1$, $\gamma_2$, $\dots$, $\gamma_n$ satisfy relations similar to that satisfied by the weight matrices that we need in the first row (squaring to $-1$ and anticommuting). Similarly the symbols $\delta_k,k=1,\dots,a$, $\bigcup_{m=2}^{a}\prod_{i=1}^{m}\delta_{k_i}|1\leq k_i\leq k_{i+1}\leq a$ satisfy relations similar to that satisfied by the weight matrices that we need in the first column (squaring to $1$ and commuting with all other elements). Thus when the weight matrices of any CUW code are expressed in the array form as discussed in the previous section, the matrices in the first row will simply be matrix representations of the symbols $1$, $\gamma_1$, $\gamma_2$, $\dots$, $\gamma_n$ of a Extended Clifford Algebra. Similarly the matrices in the first column are nothing but matrix representation of the symbols $\delta_k,k=1,\dots,a$, $\bigcup_{m=2}^{a}\prod_{i=1}^{m}\delta_{k_i}|1\leq k_i\leq k_{i+1}\leq a$ of a Extended Clifford Algebra. Thereby the entire problem as described in the previous section can now be restated in algebraic terms as follows.

\textit{What is the minimum matrix size $N_t$ in which the algebra $\mathbb{A}_{(g-1)}^\lambda$ has a non-trivial matrix representation?}
 
\section{Maximum rate of CUW codes}
\label{sec4}
In this section, the maximum rate of CUW codes is found using tools from representation theory for $\lambda=2^a,a\in\mathbb{N}$. Then, an algebraic explanation is given for the tensor product based CUW code construction in \cite{KaR2}.

As discussed in the previous section, the problem is to find the minimum dimension in which we can have a representation of the algebra $\mathbb{A}_{(g-1)}^\lambda$. But this problem appears to be difficult to solve directly. Hence we take an alternate approach which is similar to the approach in \cite{TiH} wherein matrix representations of Clifford algebras was obtained using matrix representations of the Clifford group. First we find a finite group with respect to multiplication in the algebra $\mathbb{A}_{(g-1)}^{\lambda}$ such that it contains the elements of the natural $\mathbb{R}$-basis of $\mathbb{A}_{(g-1)}^{\lambda}$ denoted by $\mathscr{B}_{(g-1)}^{\lambda}$. Then we find a suitable representation of this finite group such that it can be extended to a representation of the algebra.
\begin{proposition}
The set of elements
\begin{equation}
G=\mathscr{B}_{(g-1)}^{\lambda}\cup\left\{-b|b\in\mathscr{B}_{(g-1)}^{\lambda} \right\}
\end{equation}
is a finite group with respect to multiplication in $\mathbb{A}_{(g-1)}^{\lambda}$. Further, the group $G$ is a direct product of its subgroups $G_{\mu}$ and $G_{\delta}$, where
\begin{equation}
\begin{array}{rcl}
G_{\mu}&=&\mathscr{B}_{(g-1)}\cup\left\{-b|b\in\mathscr{B}_{(g-1)}\right\},\\
G_{\delta}&=&\left\{1,\delta_1\right\}\times\left\{1,\delta_2\right\}\times\dots\times\left\{1,\delta_a\right\}.
\end{array}
\end{equation}
\end{proposition}
\begin{proof}
The multiplication is associative and the unit is $1$. The inverse of the element $\pm\prod_{i=1}^{m}\gamma_{k_i}$ is $\pm(-1)^{\lceil\frac{m}{2}\rceil}\prod_{i=1}^{m}\gamma_{k_i}$. The inverse of the element $\prod_{i=1}^{m}\delta_{k_i}$ is itself. Similarly, it is easy to find the inverse of the other elements. The set $G_{\mu}$ is nothing but the well known Clifford group \cite{TiH}. The set $G_\delta$ is a group obviously, since it is a direct product of the cyclic group $C_2$, $a$ number of times. The group $G$ is a direct product of $G_{\mu}$ and $G_{\delta}$ because of the following reasons.
\begin{enumerate}
\item Each $s\in G$ can be written uniquely in the form $s=s_1s_2$ with $s_1\in G_{\mu}$ and $s_2\in G_{\delta}$.
\item For $s_1\in G_{\mu}$ and $s_2\in G_{\delta}$, we have $s_1s_2=s_2s_1$.
\end{enumerate} 
\end{proof}

Thus, the problem is simplified to finding the matrix representation of this finite group $G$. To start, we quickly recall some basic concepts in linear representation of finite groups. We refer the readers to \cite{Serre} for a formal introduction.
\begin{defn}\cite{Serre}
Let $G$ be a finite group with identity element $1$ and let $V$ be a finite dimensional vector space over $\mathbb{C}$. A linear representation of $G$ in $V$ is a group homomorphism $\rho$ from $G$ into the group $GL(V)$. The dimension of $V$ is called the degree of the representation.
\end{defn}
\begin{enumerate}
\item[R1:] Irreducible representations are representations with no invariant subspaces. 
\item[R2:] Every representation is a direct sum of irreducible representations. They are equivalent to block-diagonal representations, with irreducible representation matrices on the block diagonal.
\item[R3:] Two representations $R$ and $R'$ of $G$ are equivalent, if there exists a similarity transform $U$ so that
$$
R'(x)=U^{-1}R(x)U,~\forall~x\in G
$$
\item[R4:] Unitary group representations are representations in terms of unitary matrices
\item[R5:] Every representation is equivalent to a unitary representation    
\end{enumerate}
\begin{theorem}\cite{Serre}
\label{thm_abelian}
All the irreducible representations of an Abelian group have degree $1$.
\end{theorem}
\begin{lemma}\cite{Serre}
\label{lem_tensor}
Let $\rho_1:G_1\to GL(V_1)$ and $\rho_2:G_2\to GL(V_2)$ be linear representations of groups $G_1$ and $G_2$ in vector spaces $V_1$ and $V_2$ respectively. Then $\rho_1\otimes\rho_2$ is a linear representation of $G_1\times G_2$ into $V_1\otimes V_2$.
\end{lemma}
\begin{theorem}\cite{Serre}
\label{thm_tensor}
\begin{enumerate}
\item If $\rho_1$ and $\rho_2$ are irreducible, then $\rho_1\otimes\rho_2$ is an irreducible representation of $G_1\times G_2$.
\item Each irreducible representation of $G_1\times G_2$ is equivalent to a representation $\rho_1\otimes\rho_2$, where $\rho_i$ is an irreducible representation of $G_i,~i=1,2.$ 
\end{enumerate}
\end{theorem}

Now that we have introduced the necessary tools required, the problem is to find unitary matrix representations of the finite group $G$. Before we proceed, note that when $G$ is interpreted as a finite group, the representation of $-1$ does not necessarily have anything to do with $-1$ times identity matrix and similarly for a generic $-b,b\in\mathscr{B}_{(g-1)}^{\lambda}$. Such a representation $\rho$, where $\rho(-1)\neq-\rho(1)$ is said to be a degenerate representation. Degenerate representations are not representations of the algebra $\mathbb{A}_{(g-1)}^{\lambda}$. Thus we are interested in a non-degenerate unitary representation $\rho$ of the finite group $G$ such that the following conditions are satisfied.
\begin{enumerate}
\item $\rho(x)\neq\pm I,~ \forall~ x\in G,x\neq\pm 1$.
\item $\rho(x)\neq\pm\rho(y),~\forall~x\neq y\in G_{\delta}$ 
\item The degree of representation should be as small as possible.
\end{enumerate}

The first two conditions are required for \textit{unique decodability} \cite{KaR1}, since otherwise there will be problems in decoding for certain choice of signal sets. This issue is elaborated in \cite{KaR1} and hence not repeated here. 
\begin{theorem}
\label{thm_maxrate}
The maximum rate (in complex symbols per channel use) of a $\lambda$-real symbol ML decodable CUW code in $K=g\lambda$ real variables for $\lambda=2^a,a\in\mathbb{N}$ is given by
$$
R_{max}=\frac{g}{2^{\left(\lfloor\frac{(g-1)}{2}\rfloor+1\right)}} 
$$
\end{theorem}
\begin{proof}
Let us first consider the case $\lambda=2$ and arbitrary $g$. Then the finite group $G=G_{\mu}\times G_{\delta}$, where $G_{\delta}=\left\{1,\delta_1\right\}$. Since we are interested in minimizing the degree of representation, we first study the irreducible representations of $G$. From Theorem \ref{thm_tensor}, we know that all the irreducible representations of $G$ are obtained as tensor product of the irreducible representations of $G_{\mu}$ and $G_{\delta}$. All the irreducible representations of $G_{\mu}$ are available in \cite{TiH}. The non-degenerate irreducible representations of $G_{\mu}$ are also available in \cite{TiH} for dimension $2^{\lfloor\frac{g-1}{2}\rfloor}$. Let us denote the non-degenerate representation of $G_{\mu}$ by $\rho_1$. The group $G_{\delta}$ is Abelian. Thus, from Theorem \ref{thm_abelian}, all the irreducible representations of $G_{\delta}$ are in dimension $1$. Apart from the trivial representation (all elements are mapped to $1$), there is only one irreducible representation $\rho_2$ of $G_{\delta}$ given by 
$$
\rho_2(1)=1,\quad\rho_2(\delta_1)=-1
$$
Thus, we have explicitly obtained all the irreducible representations of $G$ for the case $\lambda=2$. They are in dimension $2^{\lfloor\frac{g-1}{2}\rfloor}$. However, all the irreducible representations of $G$ fail to satisfy the required conditions. Thus we seek reducible representations of $G$. Reducible representations can be easily constructed by placing irreducible representations on the blocks of the diagonal. By doing so, it can be checked that the smallest dimension reducible representation $\rho$ of $G$ satisfying the requirements is $2\left(2^{\lfloor\frac{g-1}{2}\rfloor}\right)$. It is given explicitly as follows.
\begin{equation}
\begin{array}{l}
\rho(1)=I_{2m},~\rho(\delta_1)=\left[\begin{array}{cc}I_m & 0\\0 & -I_m\end{array}\right]\\
\rho(\gamma_i)=\left[\begin{array}{cc}\rho^1(\gamma_i) & 0\\0 & \rho^1(\gamma_i)\end{array}\right],~ i=1,\dots,(g-1)
\end{array}
\end{equation}
where, $m=2^{\lfloor\frac{g-1}{2}\rfloor}$.
By applying the same arguments as in the case of $\lambda=2$ repeatedly, it can be shown that the smallest degree of representation satisfying the requirements for $\lambda=2^a,a\in\mathbb{N}$ is $\lambda\left(2^{\lfloor\frac{g-1}{2}\rfloor}\right)$. Thus the maximum rate is given by $R_{max}=\frac{(\frac{g\lambda}{2})}{\lambda\left(2^{\left\lfloor\frac{g-1}{2}\right\rfloor}\right)}=\frac{g}{2^{\left(\left\lfloor\frac{g-1}{2}\right\rfloor+1\right)}} $     
\end{proof}
Note that the expression for $R_{max}$ is independent of $\lambda$. Moreover, the construction of maximal rate CUW code for $\lambda=2$ in the proof of Theorem \ref{thm_maxrate} can be easily generalized for any $\lambda=2^a,a\in\mathbb{N}$. Observe that such a construction leads to weight matrices with a block diagonal structure. This is because we use reducible representation of groups. 

\subsection{Algebraic explanation for Tensor product construction}
In \cite{KaR2}, a construction of CUW codes based on tensor products was provided without giving any reasoning for the mathematical source of such a construction. With the algebraic background that we have now developed, the tensor product construction in \cite{KaR2} can be easily explained. Since the group $G$ is a direct product of $G_\mu$ and $G_{\delta}$, from Lemma \ref{lem_tensor}, a representation of $G$ can be obtained as a tensor product of a representation of $G_{\mu}$ and that of $G_{\delta}$. The unitary matrix representation of $G_{\mu}$ is available in \cite{TiH}. The unitary  matrices representing $G_{\delta}$ should commute and also square to $I$. Such matrices can be constructed easily, since they are simultaneously diagonalizable and their eigen values are equal to $\pm 1$ (squaring to $I$). The construction suggested in \cite{KaR2} is precisely based on this principle.

\section{Algebraic explanation for ABBA construction}
\label{sec5}
In this section, an algebraic explanation is given for the 'ABBA' construction \cite{TBH} based on the concept of right module algebra over non-commutative rings. 

As illustrated in Section \ref{sec3}, the algebra $\mathbb{A}_n^{2^a}$ over $\mathbb{R}$ can also be viewed as a finitely generated right module algebra over $Cliff_n$. Let $L=2^a$. Then a general element $x$ of the algebra $\mathbb{A}_n^L$ can be written as follows.
\begin{equation}
x=c_1+\delta_1c_2+\dots+\delta_ac_{a+1}+\delta_1\delta_2c_{a+2}+\dots+(\prod_{i=1}^{a}\delta_i)c_{L}
\end{equation}
where $c_i,i=1,\dots,L~\in Cliff_n$. There is a natural embedding of $\mathbb{A}_n^L$ into $End_{Cliff_n}(\mathbb{A}_n^L)$ given by left multiplication as follows.
\begin{equation}
\begin{array}{l}
\phi:\mathbb{A}_n^L\mapsto \mathrm{End}_{Cliff_n}(\mathbb{A}_n^L),\\
\phi(x)=L_x:y\mapsto xy.
\end{array}
\end{equation}
It is easy to check that the map $L_x$ is $Cliff_n$ linear and the map $\phi$ is a ring homomorphism. Hence, we can represent the algebra $\mathbb{A}_n^L$ by matrices with entries from Clifford algebra. However, we are only interested in matrix representations with entries from the complex field. But this can be easily obtained by simply replacing each Clifford algebra element by its matrix representation over $\mathbb{C}$. This is possible because the matrix representation of $Cliff_n$ over $\mathbb{C}$ is well known and is explicitly given in \cite{TiH}. We now illustrate this with an example.
\begin{eg}
Consider $\mathbb{A}_n^{2^a}$ for $a=2$. A general element $x\in \mathbb{A}_n^4$ can be expressed as follows.
$$
x=c_1+\delta_1c_2+\delta_2c_3+\delta_1\delta_2c_4
$$
where, $c_i,i=1,\dots,4 \in Cliff_n$. Let us now obtain a matrix representation over $Cliff_n$ for the map $L_x$. We have,
\begin{equation}
\begin{array}{rcl}
L_x(1)&=&c_1+\delta_1c_2+\delta_2c_3+\delta_1\delta_2c_4\\
L_x(\delta_1)&=&(c_1+\delta_1c_2+\delta_2c_3+\delta_1\delta_2c_4)\delta_1\\
&=&\delta_1c_1+c_2+\delta_1\delta_2c_3+\delta_2c_4\\
L_x(\delta_2)&=&(c_1+\delta_1c_2+\delta_2c_3+\delta_1\delta_2c_4)\delta_2\\
&=&\delta_2c_1+\delta_1\delta_2c_2+c_3+\delta_1c_4\\
L_x(\delta_1\delta_2)&=&(c_1+\delta_1c_2+\delta_2c_3+\delta_1\delta_2c_4)\delta_1\delta_2\\
&=&\delta_1\delta_2c_1+\delta_2c_2+\delta_1c_3+c_4. 
\end{array}
\end{equation}
\end{eg}
Thus the map $L_x$ can be represented as the following matrix
\begin{equation}
\left[\begin{array}{cccc}
c_1 & c_2 & c_3 & c_4\\
c_2 & c_1 & c_4 & c_3\\
c_3 & c_4 & c_1 & c_2\\
c_4 & c_3 & c_2 & c_1
\end{array}\right]
\end{equation}
where, $c_1,c_2,c_3,c_4\in Cliff_n$. Now to get a matrix representation over $\mathbb{C}$, we simply replace each $c_i,i=1,\dots,4$ by their matrix representations over $\mathbb{C}$. However, to get a $\lambda$-real symbol ML decodable code, we are interested only in the linear design obtained using matrix representation of the specific elements $\delta_k,k=1,\dots,a$, $\bigcup_{m=2}^{a}\prod_{i=1}^{m}\delta_{k_i}|1\leq k_i\leq k_{i+1}\leq a$, $1$, $\gamma_1$, $\gamma_2$, $\dots$, $\gamma_n$ of the algebra. This can be obtained by simply restricting the representation to the subspace over $\mathbb{R}$ generated by the required elements of the algebra. In other words, we put zero for the coefficients corresponding to the terms not required. Thus, we simply replace $c_i,i=1,\dots,4$ by CODs. Hence, we obtain a $\lambda$-real symbol ML decodable CUW code with maximal rate. It turns out that the above construction is precisely the ABBA construction proposed by Tirkkonen et al in \cite{TBH}.

As a consequence of this result, it follows that the $4$ transmit antenna linear design based on ABBA construction shown below is a $2$-real symbol ML decodable code
\begin{displaymath}
\left[\begin{array}{crcr}
x_1 & -x_2^* & x_3 & -x_4^*\\
x_2 & x_1^* & x_4 & x_3^*\\
x_3 & -x_4^* & x_1 & -x_2^*\\
x_4 & x_3^* & x_2 & x_1^*
\end{array}\right].
\end{displaymath}

Though the same linear design was proposed earlier in \cite{TBH}, the authors of \cite{TBH} chose the following pairing of real variables which essentially made the linear design into a $4$-real symbol ML decodable code.
\begin{enumerate}
\item First group $\left\{x_{1I},x_{1Q}\right\}$ 
\item Second group $\left\{x_{2I},x_{2Q}\right\}$
\item Third group $\left\{x_{3I},x_{3Q}\right\}$
\item Fourth group $\left\{x_{4I},x_{4Q}\right\}$.
\end{enumerate}
However, if we form the following partition of real variables, we can obtain a single complex symbol ML decodable code. 
\begin{enumerate}
\item First group $\left\{x_{1I},x_{3I}\right\}$ 
\item Second group $\left\{x_{1Q},x_{3Q}\right\}$
\item Third group $\left\{x_{2I},x_{4I}\right\}$
\item Fourth group $\left\{x_{2Q},x_{4Q}\right\}$.
\end{enumerate}
Note that the pair of real variables in one group should be allowed to take values independently of the real variables in other groups. For example, the pair of real variables $x_{1I},x_{3I}$ should take values jointly from a two dimensional constellation independently of the real variables in other groups. Thus we see that the ML decoding complexity of STBCs obtained from linear designs can vary dramatically depending on the partitioning of real variables into groups and the choice of signal sets.  

\section{Discussions}
\label{sec6}
The CUW codes \cite{KaR1,KaR3,KaR2} are based on sufficient conditions for $g$-group ML decodability. An algebraic framework for $g$-group ML decodable codes based on the necessary and sufficient conditions and the maximal rate of general $g$-group ML decodable codes are currently under investigation. 
\section*{Acknowledgment}
This work was supported through grants to B.S.~Rajan; partly by the
IISc-DRDO program on Advanced Research in Mathematical Engineering, and partly
by the Council of Scientific \& Industrial Research (CSIR, India) Research
Grant (22(0365)/04/EMR-II). The first author would like to thank Prof. D.P. Patil for exposing him to the area of commutative algebra through special courses and also for suggesting \cite{Serre} as a good reference for representation theory.




\begin{thebibliography}{1}
\bibitem{TJC} V. Tarokh, H. Jafarkhani and A.R. Calderbank, ``Space-Time Codes from Orthogonal Designs,'' \emph{IEEE Trans. Inform. Theory}, Vol. 45, No. 2, pp. 744-765.

\bibitem{TiH} Olav Tirkkonen and Ari Hottinen, ``Square-Matrix Embeddable Space–Time Block Codes for Complex Signal Constellations," \emph{IEEE Trans. Inform. Theory}, Vol. 48, No. 2, pp. 384-395, Feb. 2002.

\bibitem{TBH} Olav Tirkkonen, Adrian Boariu and Ari Hottinen, ``Minimal Non-Orthogonality Rate 1 Space-Time Block Code for $3+$ Tx Antennas,'' Proceedings of \emph{IEEE Int. Symp. on Spread-Spectrum Tech. \& Applications.}, New Jersey, Sept.6-8, 2000, pp. 429-432.  

\bibitem{Jaf} H. Jafarkhani, ``A quasi-orthogonal space-time block code,'' \emph{IEEE Trans. on Communications}, Vol. 49, No. 1, pp.1-4, Jan. 2001.

\bibitem{KhR} Md. Zafar Ali Khan and B. Sundar Rajan, ``Single-Symbol Maximum-Likelihood Decodable Linear STBCs," \emph{IEEE Trans. Inform. Theory}, Vol.52, No.5, pp.2062-2091, May 2006.

\bibitem{KaR1} Sanjay Karmakar and B.Sundar Rajan, ``Minimum-Decoding Complexity, Maximum-rate Space-Time Block Codes from Clifford Algebras," Proc. \emph{IEEE Int. Symp. on Inform. Theory}, Seattle, July 9-14, 2006, pp.788-792.

\bibitem{KaR3} Sanjay Karmakar, B.Sundar Rajan, ``High-rate Multi-Symbol-Decodable STBCs from Clifford Algebras,'' Proceedings of \emph{13th National Conference on Communications(NCC 2007)}, IIT Kanpur, Jan.26-28, 2007. Available in arXiv:cs.IT/0702023.

\bibitem{KaR2} Sanjay Karmakar and B. Sundar Rajan, ``Multi-group decodable STBCs from Clifford Algebras,'' Proceedings of IEEE Inform. Theory Workshop (ITW 2006), Chengdu, China, October 22-26, 2006, pp.448-452.

\bibitem{YGT} Chau Yuen, Yong Liang Guan, Tjeng Thiang Tjhung, ``A class of four-group quasi-orthogonal space-time block code achieving full rate and full diversity for any number of antennas," in Proceedings \emph{IEEE Personal, Indoor and Mobile Radio Communications Symposium}(PIMRC), Vol.1, pp.92 - 96, Berlin, Germany, Sept.11-14, 2005.

\bibitem{RaR_isit_p3} G. Susinder Rajan and B. Sundar Rajan, ``Algebraic Distributed Space-Time Codes with Low ML Decoding Complexity,'' to appear in Proceedings of ISIT 2007, Nice, France, June 24-29, 2007.

\bibitem{RaR_isit_p2} G. Susinder Rajan and B. Sundar Rajan, ``Noncoherent Low-Decoding-Complexity Space-Time Codes for Wireless Relay Networks,'' to appear in Proceedings of ISIT 2007, Nice, France, June 24-29, 2007.

\bibitem{Jac} N.Jacobson, Basic Algebra I, 2nd ed., New York:W.H.Freeman, 1985.

\bibitem{Serre} Jean-Pierre Serre, ``Linear Representations of Finite Groups,'' Springer-Verlag.

\end{thebibliography}
%

\end{document}